\documentclass[10pt,journal,compsoc]{IEEEtran}
\ifCLASSOPTIONcompsoc
  \usepackage[nocompress]{cite}
\else
  \usepackage{cite}
\fi

\ifCLASSINFOpdf
  \usepackage[pdftex]{graphicx}
\else
\fi

\usepackage{amsmath}
\usepackage{amssymb}

\usepackage{tcolorbox}
\usepackage{colortbl}
\usepackage{bbm}
\usepackage{bbding}
\usepackage{dutchcal}
\usepackage{url}

\usepackage{amsthm}

\newtheorem{definition}{Definition}
\newtheorem{theorem}{Theorem}
\newtheorem{lemma}[theorem]{Lemma}
\newtheorem{proposition}[theorem]{Proposition}

\newcommand{\E}{\mathbb{E}}
\newcommand{\pdf}{\operatorname{pdf}}
\newcommand{\var}{\operatorname{Var}}
\newcommand{\cov}{\operatorname{Cov}}
\newcommand{\cdf}{\operatorname{cdf}}
\newcommand{\hp}{\bar{\mathbcal{p}}}
\DeclareMathOperator{\hr}{\bar{\mathbcal{r}}}
\DeclareMathOperator{\hhq}{\bar{\mathbcal{q}}}
\newcommand{\prob}{\mathbb{P}}

\definecolor{LightCyan}{rgb}{0.88,1,1}
\definecolor{OliveGreen}{rgb}{0,0.6,0}

\usepackage{algpseudocode}

\ifCLASSOPTIONcompsoc
  \usepackage[caption=false,font=footnotesize,labelfont=sf,textfont=sf]{subfig}
\else
  \usepackage[caption=false,font=footnotesize]{subfig}
\fi

\usepackage{url}

\begin{document}
\title{HyperMinHash: MinHash in LogLog space}

\author{Yun~William~Yu,
        Griffin~M.~Weber
\IEEEcompsocitemizethanks{\IEEEcompsocthanksitem Y. Yu is with the Department of Mathematics,
University of Toronto, Toronto, Ontario, Canada M5S 2E4.
G. Weber is
with the Department of Biomedical Informatics, Harvard Medical School, Boston, MA 02115
\protect\\
E-mail: ywyu@math.toronto.edu and griffin\_weber@hms.harvard.edu}
\thanks{Manuscript received July 6, 2018; revised July 12, 2019}
}

\markboth{Journal of \LaTeX\ Class Files,~Vol.~14, No.~8, August~2015}%
{Shell \MakeLowercase{\textit{et al.}}: Bare Demo of IEEEtran.cls for Computer Society Journals}

\IEEEtitleabstractindextext{%
\begin{abstract}
    In this extended abstract, we describe and analyze a lossy compression of MinHash from buckets of size $O(\log n)$ to buckets of size $O(\log\log n)$ by encoding using floating-point notation.
    This new compressed sketch, which we call HyperMinHash, as we build off a HyperLogLog scaffold, can be used as a drop-in replacement of MinHash.
    Unlike comparable Jaccard index fingerprinting algorithms in sub-logarithmic space (such as b-bit MinHash), HyperMinHash retains MinHash's features of streaming updates, unions, and cardinality estimation.
    For a multiplicative approximation error $1+ \epsilon$ on a Jaccard index $ t $, given a random oracle, HyperMinHash needs $O\left(\epsilon^{-2} \left( \log\log n + \log \frac{1}{ t  \epsilon} \right)\right)$ space.
    HyperMinHash allows estimating Jaccard indices of 0.01 for set cardinalities on the order of $10^{19}$ with relative error of around 10\% using 64KiB of memory; MinHash can only estimate Jaccard indices for cardinalities of $10^{10}$ with the same memory consumption.
\end{abstract}

}

\maketitle

\IEEEdisplaynontitleabstractindextext

\IEEEpeerreviewmaketitle

\IEEEraisesectionheading{\section{Introduction}\label{sec:introduction}}

\IEEEPARstart{M}{any} questions in data science can be rephrased in terms of the number of items in a database that satisfy some Boolean formula, or by the similarity of multiple sets as measured through the relative overlap.
For example, \textit{``how many participants in a political survey are independent and have a favorable view of the federal government?''}, or \textit{``how similar are the source IPs used in a DDoS attack today vs.\ last month?''}
The MinHash sketch, developed by Broder in 1997 \cite{broder1997resemblance}, is standard technique used by industry to answer these types of questions \cite{li2005using}.
MinHash uses $O(\epsilon^{-2} \log n)$ space, where $n$ is the number of unique items,
allows streaming updates, and directly estimates both the union cardinality and Jaccard index \cite{jaccard1902lois}.
Further, by coupling together Jaccard index with union cardinality, intersection sizes are also accessible.

The literature provides sub-logarithmic near-optimal probabilistic data structures for approximating the `count-distinct' problem 
\cite{bar2002counting, durand2003loglog, flajolet2004counting, flajolet2007hyperloglog, kane2010optimal}.
Because these sketches can be losslessly merged to find the sketch of the union of sets, this enables determining the cardinality of unions of sets.
Rephrased as a query on items satisfying a Boolean formula, unions enable OR queries.
Similarly, sub-logarithmic near-optimal fingerprints exist for computing Jaccard index \cite{jaccard1902lois, bachrach2010fast, li2010b, bachrach2015fingerprints}.
When paired with a count-distinct sketch, this allows AND queries by looking at the intersection cardinalities of sets.
Unfortunately, in general, Jaccard index fingerprints (such as b-bit MinHash \cite{li2010b}) cannot be merged to form the Jaccard index of a union of sets.
While queries of the form $|A \cap B|$ for sets $A$ and $B$ can be performed by combining together a count-distinct sketch and a Jaccard index fingerprint, more complicated queries, such as $|A \cap (B \cup C)|$ cannot.
One of the major advantages of MinHash is that as a merge-able sketch providing both count-distinct and Jaccard index estimators, it can be composed in this fashion.

To our knowledge, there are no practical sub-logarithmic sketches to replace MinHash in that problem space, because existing sub-logarithmic Jaccard index fingerprints cannot be merged.
For this reason, much work has gone into improving the performance characteristics of MinHash sketches.
These advances include requiring only a single permutation \cite{li2012one} and showing that very minimal independence is actually needed in the random hash function \cite{thorup2013bottom, feigenblat2011exponential, feigenblat2017dk}.
Here, we build on prior work and use a LogLog counter to give a lossy compression of the minimum hash values that reduces the required storage size for each hash from $O(\log n)$ to $O(\log \log n)$.

\subsection{MinHash}
Given two sets $A$ and $B$, where $|A| = n$ and $|B| = m$, and $n > m$, the Jaccard index is defined as
\begin{equation}
t(A,B) = \frac{|A \cap B|}{|A \cup B|}.
\end{equation}
Clearly, if paired with a good count-distinct estimator for $|A \cup B|$, this allows us to estimate intersection sizes as well.
Though Jaccard originally defined this index to measure ecological diversity in 1902 \cite{jaccard1902lois}, in more modern times, it has been used as a proxy for the document similarity problem.
In 1997, Broder introduced \textit{min-wise hashing} (colloquially known as `MinHash') \cite{broder1997resemblance}, a technique for quickly estimating the resemblance of documents by looking at the Jaccard index of `shingles' (collections of phrases) contained within documents.

MinHash Jaccard-index estimation relies on a simple fact: if you apply a random permutation to the universe of elements, the chance that the smallest item under this permutation in sets $A$ and $B$ are the same is precisely the Jaccard index.
To see this, consider a random permutation of $A \cup B$.
The minimum element will come from one of three disjoint sets: $A \setminus B$, $B \setminus A$, or $A \cap B$.
If the minimum element lies in $A \setminus B$, then $\min(A) \not \in B$, so $\min(A) \ne \min(B)$; the same is of course true by symmetry for $B \setminus A$.
Conversely, if $\min(A \cup B) \in A \cap B$, then clearly $\min(A) = \min(B)$.
Because the permutation is random, every element has an equal probability of being the minimum, and thus
\begin{equation}
\prob\left(\min(A) = \min(B)\right) = \frac{|A \cap B|}{|A \cup B|}.
\end{equation}

While using a single random permutation produces an unbiased estimator of $t(A,B)$, it is a Bernouli 0/1 random variable with high variance.
So, instead of using a single permutation, one can average $k$ trials.
The expected fraction of matches is also an unbiased estimator of the Jaccard index, but with variance decreased by a multiplicative factor of $1/k$.

Though the theoretical justification is predicated on having a true random permutation, in practice we approximate that by using hash functions instead.
A good hash function will specify a nearly total ordering on the universe of items, and provided we use $\theta(\log(n))$ bits for the hash function output space, the probability of accidental collision is exponentially small.

Though theoretically easy to analyze, this scheme has a number of drawbacks, chief amongst them the requirement of having $k$ random hash functions, implying a $\theta(nk)$ computational complexity for generating the sketch.
To address this, several variants of MinHash have been proposed \cite{cohen2016min}:
\begin{enumerate}
    \item \textbf{k-hash functions.} The scheme described above, which has the shortcoming of using $\theta(nk)$ computation to generate the sketch.
    \item \textbf{k-minimum values.} A single hash function is used, but instead of storing the single minimum value, we store the smallest $k$ values for each set (also known as the KMV sketch \cite{bar2002counting}). Sketch generation time is reduced to $O(n \log k)$, but we also incur an $O(k\log k)$ sorting penalty when computing the Jaccard index.
    \item \textbf{k-partition.} Another one-permutation MinHash variant, k-partition stochastically averages by first deterministically partitioning a set into $k$ buckets using the first couple bits of the hash value, and then stores the minimum value within each bucket \cite{li2012one}. $k$-partition has the advantage of $O(n)$ sketch generation time and $O(k)$ Jaccard index computation time, at the cost of some difficult in the analysis.
\end{enumerate}

$k$-partition one-permutation MinHash \cite{li2012one} is significantly more efficient in both time (as it uses only one hash function, and no sorting is required to compute the Jaccard index) and space (as the bits used to the partition the set into buckets can be stored implicitly).
For this reason, it is more commonly used in practice, and whenever we refer to MinHash in this paper, we mean $k$-partition one-permutation MinHash.

MinHash sketches of $A$ and $B$ can be losslessly combined to form the MinHash sketch of $A\cup B$ by taking the minimum values across buckets.
Additionally, using order statistics, it is possible to estimate the count-distinct cardinality \cite{bar2002counting}, so we can directly estimate union cardinalities using merged sketches, intersection cardinality by multiplying Jaccard index with union cardinality, and more complex set operations by rewriting the set as an intersection of unions (e.g. $|(A \cup B) \cap C|$).

\subsection{Count-distinct union cardinality}
All of the standard variants of MinHash given in the last section use logarithmic bits per bucket in order to prevent accidental collisions (i.e. we want to ensure that when two hashes match, they came from identical elements).
However, in the related problem of cardinality estimation of unique items (the `count-distinct' problem), literature over the last several decades produced several streaming sketches that require less than logarithmic bits per bucket.
Indeed, the LogLog, SuperLogLog, and HyperLogLog family of sketches requires only $\log\log(n)$ bits per bucket by storing only the position of the first 1 bit of a uniform hash function, and for a multiplicative error $\epsilon$, use a total of $O(\epsilon^{-2} \log\log n)$ bits \cite{durand2003loglog, flajolet2004counting, flajolet2007hyperloglog}.

The analyses of these methods all originally required access to a random oracle (i.e. a truly random hash function), and storing such hash functions requires an additional $O(\log n)$ bits of space, for a total of $O(\epsilon^{-2} \log\log n + \log n)$ space.
Further compacting of the hashes by use of concentration inequalities, and use of $k$-min-wise independent hash functions allowed both relaxing the requirement of a random oracle, and reduction of the space-complexity to $O(\epsilon^{-2} + \log n)$ (or without counting the hash function, $O(\epsilon^{-2} + \log\log n )$ in the setting of shared randomness) \cite{kane2010optimal}, resulting in an essentially optimal sketch requiring only constant bits per bucket.
In practice though, the double logarithmic HyperLogLog sketch is used under the assumptions of shared randomness and that standard cryptographic hash functions (e.g. SHA-1) behave as random oracles (clearly not true, but empirically good enough)---for real-world data sets, double logarithmic is small enough to essentially be a constant $<6$.

\subsection{Jaccard index estimation}
First we note that trivially, HyperLogLog union cardinalities can be used to compute intersection cardinalities and thus Jaccard index using the inclusion-exclusion principle.
Unfortunately, the relative error is then in the size of the union (as opposed to the size of the Jaccard index for MinHash) and compounds when taking the intersections of multiple sets; for small intersections, the error is often too great to be practically feasible, unless significantly more bits are used.
More precisely, in order to achieve the same relative error, the number of bits needed when using inclusion-exclusion scales with the square of the inverse Jaccard index, as opposed to scaling with just the inverse Jaccard index (Table 1 of \cite{pagh2014min}).
Notably, some newer cardinality estimation methods based on maximum-likelihood estimation are able to more directly access intersection sizes in HyperLogLog sketches, which can then be paired with union cardinality to estimate Jaccard index \cite{ertl2017new, ertl2017new2}.
However, this approach is restricted to the information available in the HyperLogLog sketch itself, and seems empirically to be a constant order ($<3$x) improvement over conventional inclusion-exclusion.

Alternately, when unions and streaming updates are not necessary, the literature provides many examples of Jaccard index fingerprints in better than log-space \cite{bachrach2010fast, li2010b, bachrach2015fingerprints}.
State-of-the-art fingerprinting methods based on either reduction to $F_2$-norm sketching or truncating bits of MinHash (`$b$-bit MinHash') are able to achieve space bounds of $O(\epsilon^{-2})$ in the general problem \cite{bachrach2010fast, li2010b} and  $O\left(\frac{(1-t)^2}{\epsilon^{2}} \log\frac{1}{1-t}\right)$ for highly similar streams when the Jaccard index $t$ is large \cite{bachrach2015fingerprints}.
For estimating the Jaccard similarity between two sets, these fingerprinting techniques are basically asymptotically optimal \cite{pagh2014min}.

However, $b$-bit MinHash and $F_2$-norm reductions, while great for Jaccard index, lose many of the benefits of standard MinHash, even just for Jaccard index estimation.
Because $b$-bit MinHash only takes the lowest order $b$ bits of the minimum hash value \textit{after} finding the minimum, it also requires $\log(n)$ bits per bucket during the sketch generation phase, the same as standard MinHash.
This also implies a lack of mergeability: the fingerprint of the union of two sets cannot be computed from the fingerprints of the two constituent sets.
The same holds true for $F_2$-norm reductions because of double counting shared items.

\subsection{Double logarithmic MinHash}
We aim to reduce space-complexity for MinHash, while preserving all of its features; we focus on Jaccard index in the unbounded data stream model, as that is the primary feature differentiating MinHash from count-distinct sketches.
As a preliminary, we note that as with the count-distinct problem, much theoretical work has focused on reducing the amount of randomness needed by MinHash to only requiring $k$-min-wise (or $k-d$-minwise) independent hash functions \cite{feigenblat2011exponential, thorup2013bottom, feigenblat2017dk}.
In this paper, we will however again assume the setting of shared randomness and the existence of a random oracle, which is standard in industry practice.

We construct a \textit{HyperMinHash} sketch by compressing a MinHash sketch using a floating-point encoding, based off of HyperLogLog scaffold;
using HyperLogLog as a scaffold for other statistics of interest has been previously proposed, but not for MinHash \cite{cohen2017hyperloglog}.
HyperMinHash as we describe below requires
\[O\left(\epsilon^{-2} \left( \log\log n + \log \frac{1}{t \epsilon} \right)\right)\]  space and has all of the standard nice features of MinHash, including streaming updates, the ability to take unions of sketches, and count-distinct cardinality estimation.
Though this construction is not space optimal---e.g. we could very likely drop in the KNW sketch \cite{kane2010optimal} instead of HyperLogLog \cite{flajolet2007hyperloglog} for a space-reduction---it is a practical compromise and easy to implement in software.

\begin{figure*}[tb]
    \includegraphics[width=1\textwidth]{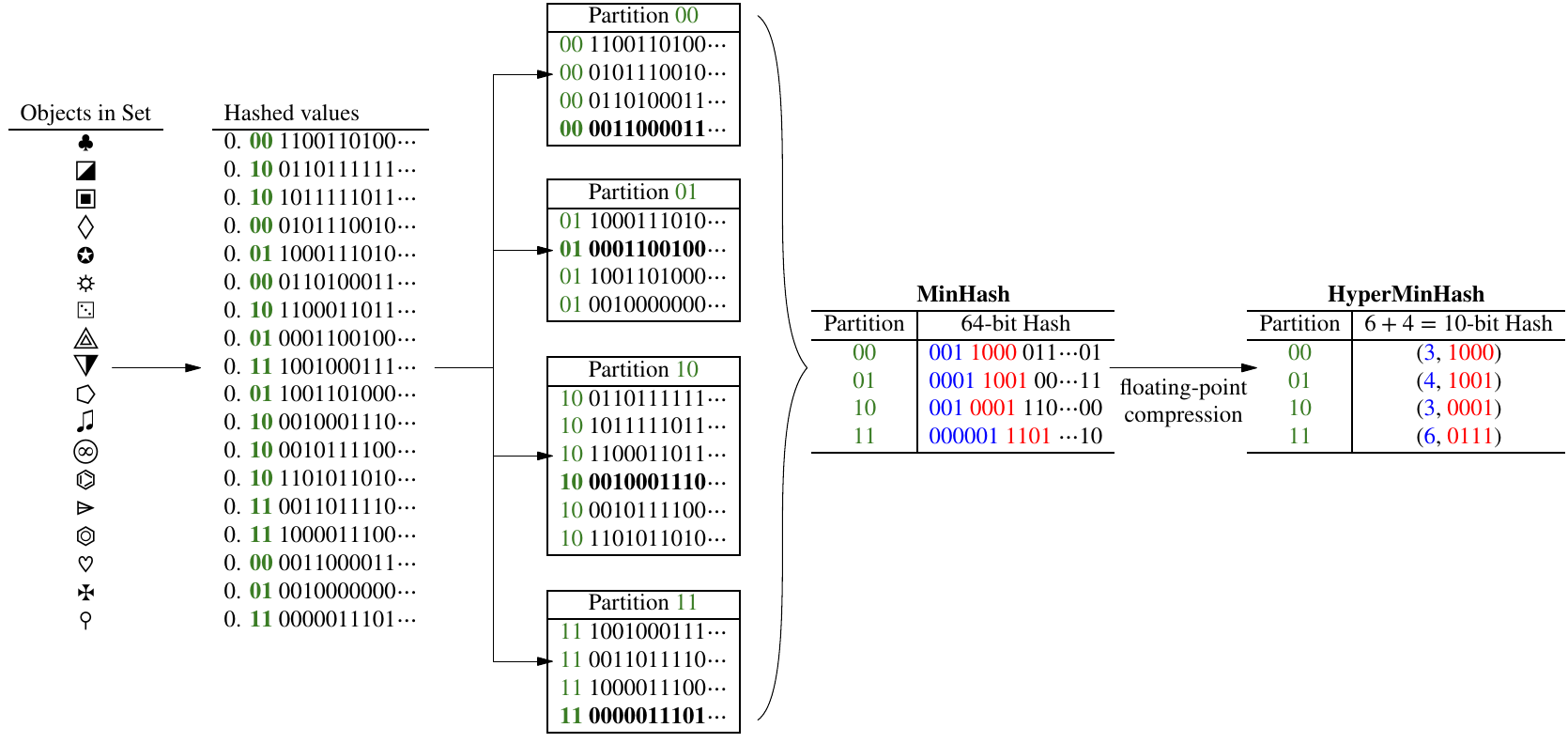}
    \caption{HyperMinHash generates sketches in the same fashion as one-permutation k-partition MinHash, but stores the hashes in floating-point notation. It begins by hashing each object in the set to a uniform random number between 0 and 1, encoded in binary. Then, the hashed values are partitioned by the first $p$ bits (above, 2 bits, in green), and the minimum value within each partition is taken. For ordinary MinHash, the procedure stops here, and a fixed number of bits (above, 64 bits) after the first $p$ bits are taken as the hash value. For HyperMinHash, each value is further lossily compressed as an exponent (blue) and a mantissa (red); the exponent is the position of the leftmost 1 bit in the next $2^q$ bits, and $2^q + 1$ otherwise (above, $q=6$). The mantissa is simply the value of the next $r$ bits in the bit-string (above, $r=4$).} 
    \label{fig:hyperminhash}
\end{figure*}

\section{Methods}
MinHash works under the premise that two sets will have identical minimum value with probability equal to the Jaccard index,
because they can only share a minimum value if that minimum value corresponds to a member of the intersection of those two sets.
If we have a total ordering on the union of both sets, the fraction of equal buckets is an unbiased estimator for Jaccard index.
However, with limited precision hash functions, there is some chance of accidental collision.
In order to get close to a true total ordering, the space of potential hashes must be on the order of the size of the union, and hence we must store $O(\log n)$ bits.

However, the minimum of a collection of uniform $[0,1]$ random variables $X_1, \ldots, X_n$ is much more likely to be a small number than a large one (the insight behind most count-distinct sketches \cite{bar2002counting}).
HyperMinHash operates identically to MinHash, but instead of storing the minimum values with fixed precision, it uses floating-point notation, which increases resolution when the values are smaller by storing an exponent and a mantissa.
We can compute the exponent of a binary fraction by taking the location of first nonzero bit in the binary expansion (the HyperLogLog part), and the mantissa is simply some fixed number of bits beyond that (the MinHash part).
More precisely, after dividing up the items into $k$ partitions, we store the position of the leading 1 bit in the first $2^q$ bits (and store $2^q + 1$ if there is no such 1 bit) and $r$ bits following that (Figure \ref{fig:hyperminhash}).
We do not need a total ordering so long as the number of accidental collisions in the minimum values is low.

To analyze the performance of HyperMinHash compared to random-permutation MinHash (or equivalently 0-collision standard MinHash) it suffices to consider the expected number of accidental collisions.
We first describe an intuitive analysis before proving things rigorously.
Here, we will also only analyze the simple case of collisions while using only a single bucket, but the same flavor of argument holds for partitioning into multiple buckets.
The HyperLogLog part of the sketch results in collisions whenever two items match in order of magnitude (Figure \ref{fig:hll-buckets}).
\begin{figure}
    \includegraphics[width=1\columnwidth]{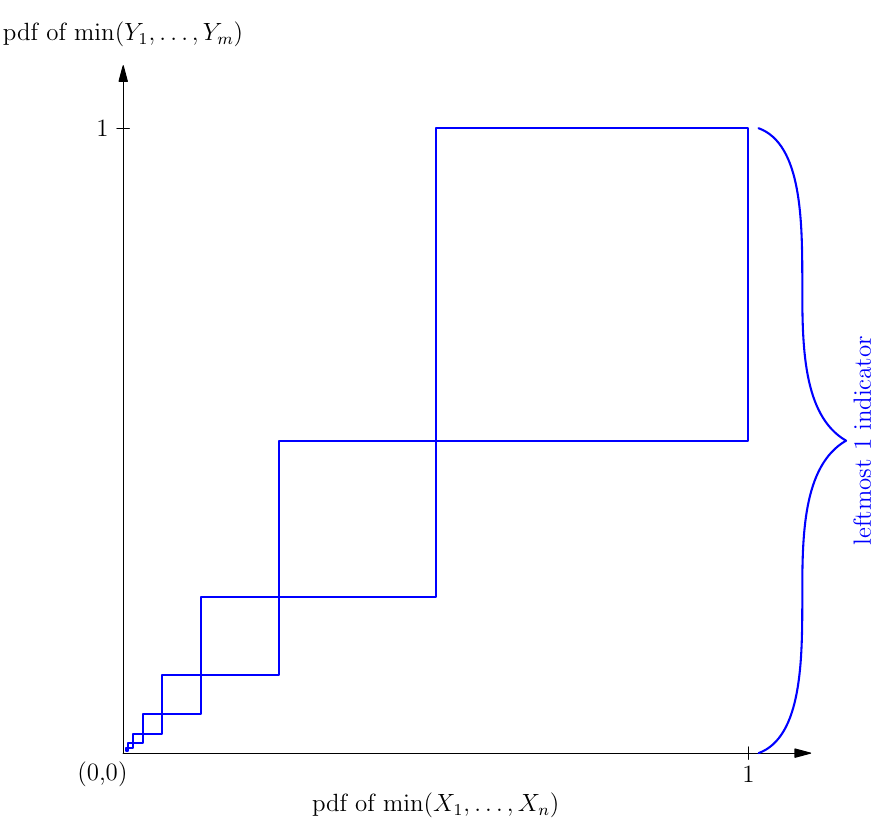}
    \caption{HyperLogLog sections, used alone, result in collisions whenever the minimum hashes match in order of magnitude.}
    \label{fig:hll-buckets}
\end{figure}

By pairing it with an addition $r$-bit hash, our collision space is narrowed by a factor of about $2^r$ within each bucket (Figure \ref{fig:hmh-buckets}).
An explicit exact formula for the expected number of collisions is
\begin{align}
\nonumber    \E C  = \sum_{i=1}^\infty \sum_{j=0}^{2^r-1} 
        & \left[\left(1-\frac{2^r+j}{2^{i+r}}\right)^n
             -\left(1-\frac{2^r+j+1}{2^{i+r}}\right)^n\right] \cdot \\ &
          \left[\left(1-\frac{2^r+j}{2^{i+r}}\right)^m
             -\left(1-\frac{2^r+j+1}{2^{i+r}}\right)^m\right] ,
\end{align}
though finding a closed formula is rather more difficult.

Intuitively, suppose that our hash value is (\textcolor{blue}{12}, \textcolor{red}{01011101}) for partition \textcolor{OliveGreen}{01}.
This implies that the original bit-string of the minimum hash was $0.\textcolor{OliveGreen}{01}\textcolor{blue}{000000000001}\textcolor{red}{01011101}\cdots$.
Then a uniform random hash in $[0,1]$ collides with this number with probability 
$2^{-(\textcolor{OliveGreen}{2} + \textcolor{blue}{12} + \textcolor{red}{8})} = 2^{-21} $.
So we expect to need cardinalities on the order of $2^{21}$ before having many collisions.
Luckily, as the cardinalities of $A$ and $B$ increase, so does the expected value of the leading 1 in the bit-string, as analyzed in the construction of HyperLogLog \cite{flajolet2007hyperloglog}.
Thus, the collision probabilities remain roughly constant as cardinalities increase, at least until we reach the precision limit of the LogLog counters.

\begin{figure}[tbh]
    \includegraphics[width=1\columnwidth]{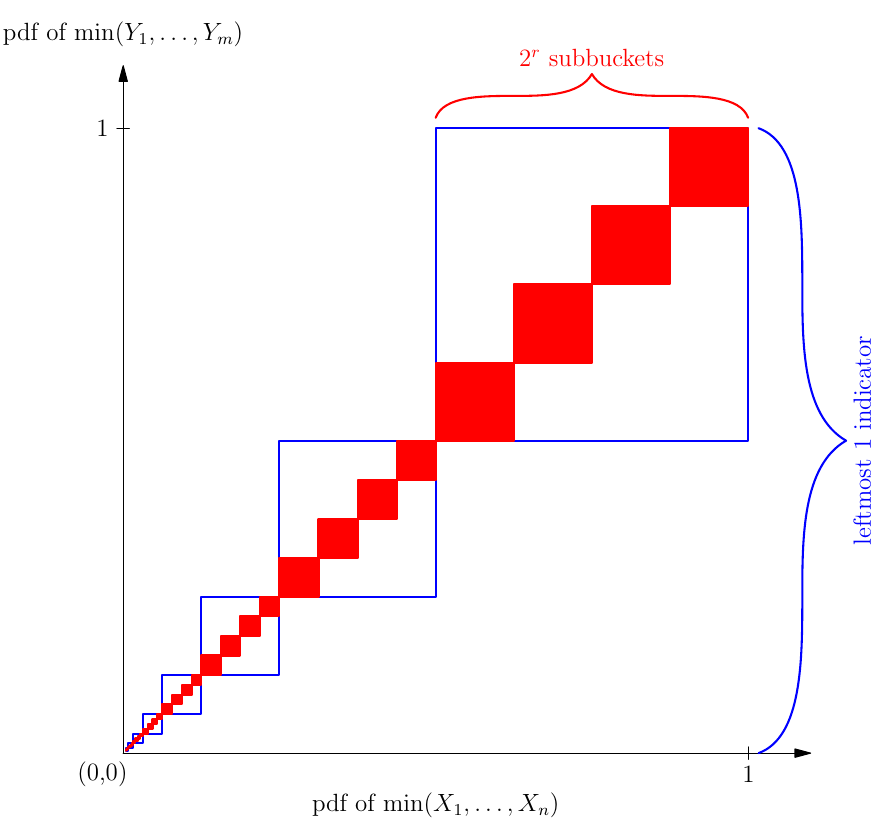}
    \caption{HyperMinHash further subdivides HyperLogLog leading 1-indicator buckets, achieving a much smaller collision space, so long as we precisely store the position of the leading 1.}
    \label{fig:hmh-buckets}
\end{figure}

But of course, we store only a finite number of bits for the leading 1 indicator (often 6 bits).
Because it is a LogLog counter, storing 6 bits is sufficient for set cardinalities up to $O(2^{2^6} = 2^{64})$.
This increases our collision surface though, as we might have collisions in the lower left region near the origin (Figure \ref{fig:hmh-in-practice}).
We can directly compute the collision probability (and similarly the variance) by summing together the probability mass in these boxes, replacing the infinite sum with a finite sum (Lemma \ref{thm:collision}).
For more sensitive estimations, we can subtract the expected number of collisions to debias the estimation.
Later, we will prove bounds on the expectation and variance in the number of collisions.

\subsection{Implementation details}
\label{app:algorithms}
Here, we present full algorithms to match a naive implementation of HyperMinHash as described above.
A Python implementation by the authors of this paper is available at \url{https://github.com/yunwilliamyu/hyperminhash}.
    Additionally, Go and Java implementations based on this pseudo-code were kindly provided by Seif Lotfy at \url{https://github.com/axiomhq/hyperminhash} and Sherif Nada at \url{https://github.com/LiveRamp/HyperMinHash-java}.
    \subsubsection{HyperMinHash sketch}
    The procedure for generating a HyperMinHash sketch can be thought of as a standard one-permutation k-partition MinHash sketch \cite{li2012one}, but where the minimum hash value is stored in a floating-point encoding.
    Programmatically though, it is easier to directly split the hash value into three parts: (1) the bucket ID, (2) the negative exponent of the bucket value, and (3) the mantissa of the bucket value.
    \label{alg:hyperminhash}
    \begin{algorithmic}[1]
        \State Let $h_1, h_2, h_3: D \to [0,1] \equiv \{0, 1\}^{\infty}$ be three independent hash functions hashing data from domain $D$ to the binary domain. (In practice, we generally use a single Hash function, e.g. SHA-1, and use different sets of bits for each of the three hashes).
        \State Let $\rho(s)$, for $s \in \{0, 1\}^\infty$ be the position of the left-most 1-bit ($\rho(0001\cdots) = 4$).
        \State Let $\sigma(s, n)$ for $s \in \{0, 1\}^\infty$ be the left-most $n$ bits of $s$ ($\sigma(01011\cdots, 5) = 01011)$.
        \Function{HyperMinHash}{$A, p, q, r$}
            \State Let $\hat{h}_1 (x) = \sigma(h_1(x), p)$.
            \State Let $\hat{h}_2(x) = \min(\rho(h_1(x)), 2^q)$.
            \State Let $\hat{h}_3(x) = \sigma(h_3(x), r)$.
            \State Initialize $2^p$ tuples $B_1 = B_2 = \cdots = B_{2^p} = (0, 0)$.
            \For {$a \in A$}
                \If {$B_{h_1(a)}[0] < \hat{h}_2(a)$}
                    \State $B_{\hat{h}_1(a)} \gets (\hat{h}_2(a), \hat{h}_3(a))$
                \Else
                    \If {$B_{\hat{h}_1(a)}[0] = \hat{h}_2(a)$ \textbf{and} $B_{\hat{h}_1(a)}[1] > \hat{h}_3(a)$}
                        \State $B_{\hat{h}_1(a)} \gets (\hat{h}_2(a), \hat{h}_3(a))$
                    \EndIf
                \EndIf
            \EndFor {}
            \State \Return $B_1, \ldots, B_{2^p}$ as $B$
        \EndFunction
    \end{algorithmic}
    \subsubsection{HyperMinHash Union}
    Given HyperMinHash sketches of two sets $A$ and $B$, we can return the HyperMinHash sketch of $A \cup B$ by for each bucket, taking the maximum exponent; or if the exponents are the same, taking the minimum mantissa. In the floating-point interpretation of the bucket values, this is simply taking the minimum bucket value.
    \label{alg:union}
    \begin{algorithmic}
        \Function{Union}{$S, T$}
            \State \textbf{assert} $|S| = |T|$
            \For {$i \in \{1, \ldots, |S|\} $}
                \State Initialize $|S|$ tuples $B_1 = B_2 = \cdots = B_{|S|} = (0, 0)$.
                \If {$S_i[0] > T_i[0]$}
                    \State $B_i \gets S_i$
                \ElsIf {$S_i[0] < T_i[0]$}
                    \State $B_i \gets T_i$
                \ElsIf {$S_i[0] = T_i[0]$}
                    \If {$S_i[[1] < T_i[1]$}
                        \State $B_i \gets S_i$
                    \Else
                        \State $B_i \gets T_i$
                    \EndIf
                \EndIf
            \EndFor
            \State \Return $B_1, \ldots, B_{2^p}$ as $B$
        \EndFunction
    \end{algorithmic}
    \subsubsection{Estimating cardinality}
    \label{alg:hllcall} Note that the left parts of the buckets can be passed directly into a HyperLogLog estimator. We can also use other $k$-minimum value count-distinct cardinality estimators, which we empirically found useful for large cardinalities.
    \begin{algorithmic}
        \Function{EstimateCardinality}{$S, p, q, r$}
            \State Initialize $|S|$ integer registers $b_1 = b_2 = \cdots = b_{|S|} = 0$.
            \For {$i \in \{1, \ldots, |S|\} $}
                \State $b_i \gets S_i[0]$
            \EndFor
            \State $R \gets \textsc{HyperLogLogCardinalityEstimator}(\{b_i\}, q)$
            \If {$ R < 1024 |S| $}
                \State \Return $R$
            \Else
                \State Initialize $|S|$ real registers $r_1, \ldots, r_{|S|}$.
                \For {$i \in \{1, \ldots, |S|\} $}
                    \State $r_i \gets 2^{-S_i[0]} \cdot \left( 1 + \frac{S_i[1]}{2^r}    \right) $
                \EndFor
                \If {$\sum r_i = 0$}
                    \State \Return $\infty$
                \Else
                    \State \Return $|S|^2 / \sum r_i$
                \EndIf
            \EndIf
        \EndFunction
    \end{algorithmic}
    \subsubsection{Computing Jaccard index}
    \label{alg:jaccard} Given two HyperMinHash sketches $A$ and $B$, we can compute the Jaccard index $t(A,B) = |A \cap B|/|A \cup B|$ by counting matching buckets. Note that the correction factor $\E C$ is generally not needed, except for really small Jaccard index. Additionally, for most practical purposes, it is safe to substitute \textsc{ApproxExpectedCollisions} for \textsc{ExpectedCollisions} (algorithms to follow).
    \begin{algorithmic}
        \Function{JaccardIndex}{$S, T, p, q, r$}
            \State \textbf{assert} $|S| = |T|$
            \State $C \gets 0$, \quad $N \gets 0$
            \For {$i \in \{1, \ldots, |S|\} $}
                \If {$S_i = T_i$ \textbf{and} $S_i \ne (0,0)$} 
                    \State $C \gets C + 1$
                \EndIf
                \If {$S_i \ne (0,0)$ \textbf{or} $T_i \ne (0,0)$}
                    \State $N \gets N + 1$
                \EndIf
            \EndFor
            \State $n \gets \textsc{EstimateCardinality}(S, q)$
            \State $m \gets \textsc{EstimateCardinality}(T, q)$
            \State $\E C \gets \textsc{[Approx]ExpectedCollisions}(n, m, p, q, r)$
            \State \Return $(C - \E C) / N$
        \EndFunction
    \end{algorithmic}%
    \subsubsection{Computing expected collisions}
    The number of expected collisions given two HyperMinHash sketches of particular sizes can be computed from Lemma \ref{thm:collision}.
    \label{alg:collisions} Note that because of floating point error, BigInts must be used for large $n$ and $m$.
    For sensitive estimation of Jaccard index, this value can be used to debias the estimator in Algorithm \ref{alg:jaccard}.
    \begin{algorithmic}
        \Function{ExpectedCollisions}{$n, m, p, q, r$}
            \State $x \gets 0$
            \For {$i \in \{1, \ldots, 2^q\} $}
                \For {$j \in \{1, \ldots, 2^r\} $}
                    \If {$i \ne 2^q$}
                        \State $b_1 \gets \frac{2^r + j}{2^{p+r+i}}$, \quad $b_2 \gets \frac{2^r + j+1}{2^{p+r+i}}$
                    \Else
                        \State $b_1 \gets \frac{j}{2^{p+r+i-1}}$, \quad $b_2 \gets \frac{j+1}{2^{p+r+i-1}}$
                    \EndIf
                    \State $Pr_x \gets (1-b_2)^n - (1-b_1)^n$
                    \State $Pr_y \gets (1-b_2)^m - (1-b_1)^m$
                    \State $x \gets x + Pr_x Pr_y$
                \EndFor
            \EndFor
            \State \Return $x \cdot 2^p$
        \EndFunction
    \end{algorithmic}
    \subsubsection{Approximating Algorithm \ref{alg:collisions}}
    \label{alg:approxcollisions}
    Here we present a fast numerically stable approximation to Algorithm \ref{alg:collisions}, which generally underestimates collisions.
    We discuss it in more detail in \ref{ssc:empirical}.
    \begin{algorithmic}
        \Function{ApproxExpectedCollisions}{$n, m, p, q, r$}
            \If{$n < m$}
                \State SWAP($x, y$)
            \EndIf
            \If{$n > 2^{2^q + r}$}
                \State \Return ERROR: cardinality too large for approximation.
            \ElsIf{$n > 2^{p+5}$}
                \State $\phi \gets \frac{4n/m}{(1+n/m)^2}$
                \State \Return $0.169919487159739093975315012348 \cdot 2^{p-r} \phi $
            \Else
                \State \Return \textsc{ExpectedCollisions}($n,m,p,q,0$) $\cdot 2^{-r}$
            \EndIf
            \State \Return $x \cdot 2^p$
        \EndFunction
    \end{algorithmic}

\subsection{Empirical optimizations}
\label{ssc:empirical}
We recommend several optimizations for practical implementations of HyperMinHash.
First, it is mathematically equivalent to:
\begin{enumerate}
    \item Pack the hashed tuple into a single word; this enables Jaccard index computation while using only one comparison per bucket instead of two.
    \item Use the $\max$ instead of $\min$ of the sub-buckets. This allows us take the union of two sketches while using only one comparison per bucket.
\end{enumerate}
These recommendations should be self-explanatory, and are simply minor engineering optimizations, which we do not use in our prototyping, as they do not affect accuracy.

However, while we can exactly compute the number of expected collisions through Lemma \ref{thm:collision}, this computation is slow and often results in floating point errors unless BigInts are used because Algorithm \ref{alg:collisions} is exponential in $r$.
In practice, two ready solutions present themselves:
\begin{enumerate}
    \item We can ignore the bias and simply add it to the error. As the bias and standard deviation of the error are the same order of magnitude, this only doubles the absolute error in the estimation of Jaccard index. For large Jaccard indexes, this does not matter.
    \item We also present a fast, numerically stable, algorithm to approximate the expected number of collisions (Algorithm \ref{alg:approxcollisions}).
\end{enumerate}

We can however approximate the number of expected collisions using the following procedure, which is empirically asymptotically correct (Algorithm \ref{alg:approxcollisions}):
\begin{enumerate}
	\item For $ n < 2^{p+5}$, we approximate by taking the number of expected HyperLogLog collisions and dividing it by $2^r$. In each HyperLogLog box, we are interested in collisions along $2^r$ boxes along the diagonal \ref{fig:hmh-in-practice}. For this approximation, we simply assume that the joint probability density function is almost uniform within the box; this is not completely accurate, but pretty close in practice.
	\item For $2^{p+5} <n < 2^{2^q + p}$, we noted empirically that the expected number of collisions approached $0.1699 \cdot 2^{p-r}$ for $n = m$ as $n \to \infty$. Furthermore, the number of collisions is dependent on $n$ and $m$ by a factor of $\frac{4nm}{(n+m)(n+m-1)}$ from \ref{thm:ray}, which for $n, m \gg 1$ can be approximated by $\frac{4n/m}{(1+(n/m)^2}$. This approximation is primarily needed because of floating point errors when $n \to \infty$.
	\item Unfortunately, around $n > 2^{2^q + p}$, the number of collisions starts increasing and these approximations fail. However, note that for reasonable values of $q = 6, p = 15$, this problem only appears when $n > 2^{89} \approx 10^{26}$.
\end{enumerate}

\begin{figure}[tb]
    \includegraphics[width=1\columnwidth]{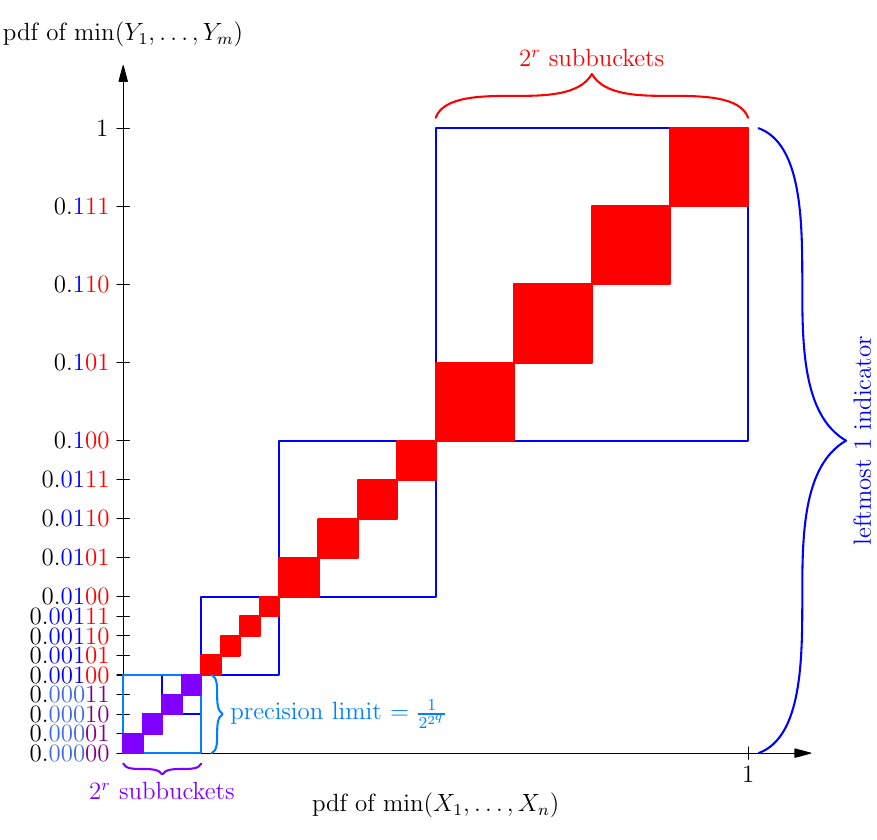}
    \caption{In practice, HyperMinHash has a limited number of bits for the LogLog counters, so there is a final lower left bucket at the precision limit.}
    \label{fig:hmh-in-practice}
\end{figure}

\subsection{Proofs}
The main result of this section bounds the expectation and variance of accidental collision, given two HyperMinHash sketches of disjoint sets.
First, we rigorously define the full HyperMinHash sketch as described above.
Note that in our proofs, we will operate in the unbounded data stream model and assume both a random oracle and shared randomness.

\begin{definition}
    We will define $f_{p,q,r}(A): \mathbb{S} \to \{\{1, \ldots, 2^q\} \times \{0, 1\}^r\}^{2^p} $ to be the HyperMinHash sketch constructed from Figure \ref{fig:hyperminhash}, where $A$ is a set of hashable objects and $p,q,r \in \mathbb{N}$, and let $f_{p,q,r}(A)_i: \mathbb{S} \to \{1, \ldots, 2^q\} \times \{0, 1\}^r$
    be the value of the $i$th bucket in the sketch.

    More precisely, let $h(x): \mathbb{S} \to [0, 1]$ be a uniformly random hash function.
    Let $\rho_q(x) = \min\left( \lfloor -\log_2(x) \rfloor + 1 , 2^q  \right) $,
    let $\sigma_r(x) = \lfloor x 2^r  \rfloor$, and
    let $\hat{h}_{q,r}(x) = \left( \rho_q(x), \sigma_r\left( x  2^{\rho_q(x)} - 1\right) \right)$.

    Then, we will define
    \begin{equation*}
        f_{p,q,r}(A)_i = \hat{h}_{q,r} \left[ \min_{\substack{a \in A \\ i2^{-p}  < h(a) < (i+1)2^{-p}}} \left(h(a) 2^p - i\right) \right].
    \end{equation*}
\end{definition}

\begin{definition}
    Let $A, B$ be hashable sets with $|A|=n$, $|B|=m$, $n>m$, and $A\cap B = \varnothing$.
    Then define an indicator variable for collisions in bucket $i$ of their respective HyperMinHash sketches
    \begin{equation}
            Z_{p,q,r}(A,B,i) = \mathbbm{1}_{\left(f_{p,q,r}(A)_i = f_{p,q,r}(B)_i \right)}.
    \end{equation}
\end{definition}

Our main theorems follow:
\begin{theorem}
    \label{thm:main}
    $C = \sum_{i=0}^{2^p-1} Z_{p,q,r}(A,B,i)$ is the number of collisions between the HyperMinHash sketches of two disjoint sets $A$ and $B$.
    Then the expectation 
    \begin{equation}
        \E C \le 2^p \left( \frac{5}{2^r} + \frac{n}{2^{p+2^q+r}}\right) .
    \end{equation}
\end{theorem}

\begin{theorem}
    \label{thm:variance}
    Given the same setup as in Theorem \ref{thm:main},
        \[  \var(C) \le \E [C]^2 + \E [C]. \]
\end{theorem}
Theorem \ref{thm:main} allows us to correct for the number of random collisions before computing Jaccard distance, and Theorem \ref{thm:variance} tells us that the standard deviation in the number of collisions is approximately the expectation.

We will first start by proving a simpler proposition.
\begin{proposition}
    Consider a HyperMinHash sketch with only 1 bucket on two disjoint sets $A$ and $B$. i.e. $f_{0,q,r}(A)$ and $f_{0,q,r}(B)$.
    Let $\gamma(n,m) \sim Z_{0,q,r}(A,B,0)$. Naturally, as a good hash function results in uniform random variables, $\gamma$ is only dependent on the cardinalities $n$ and $m$.
    We claim that
    \begin{align}
        \E \gamma(n,m) \le \frac{6}{2^r} + \frac{n}{2^{2^q+r}}. 
    \end{align}
    \label{thm:gamma}
\end{proposition}

Proving this will require a few technical lemmas, which we will then use to prove the main theorems.
\begin{lemma}
    \begin{align*}
        \E \gamma(n, m) = & \quad \prob(f_{0,q,r}(A)_0 = f_{0,q,r}(B)_0)  \\ 
=         \sum_{i=1}^{2^q-1} \sum_{j=0}^{2^r-1} 
        &         \left[\left(1-\frac{2^r+j}{2^{r+i}}\right)^n
             -\left(1-\frac{2^r+j+1}{2^{r+i}}\right)^n\right] \cdot \\ &
          \left[\left(1-\frac{2^r+j}{2^{r+i}}\right)^m
             -\left(1-\frac{2^r+j+1}{2^{r+i}}\right)^m\right] \\ 
        +  \sum^{2^q}_{i=2^q} \sum_{j=0}^{2^r-1}
        & \left[\left(1-\frac{j}{2^{r+i-1}}\right)^n
             -\left(1-\frac{j+1}{2^{r+i-1}}\right)^n\right] \cdot \\ &
          \left[\left(1-\frac{j}{2^{r+i-1}}\right)^m
             -\left(1-\frac{j+1}{2^{r+i-1}}\right)^m\right] 
    \end{align*}
    \label{thm:collision}
\end{lemma}
\begin{proof}
    Let $a_1, \ldots, a_n$ be random variables corresponding to the hashed values of items in $A$. Then $a_i \in [0,1]$ are uniform r.v. Similarly, $b_1, \ldots, b_m$, drawn from hashed values of $B$ are uniform $[0,1]$ r.v.
    Let $x = \min \{ a_1, \ldots, a_n \}$ and $y = \min \{b_1, \ldots, b_m\}$. 
    Then we have probability density functions
    \begin{align*}
        \pdf (x) = n (1-x)^{n-1}, \textrm{ for } x \in [0,1], \\
        \pdf (y) = m (1-y)^{m-1}, \textrm{ for } y \in [0,1]
    \end{align*}
    and cumulative density functions
    \begin{align*}
        \cdf(x) = 1 - (1-x)^n, \textrm{ for } x \in [0,1], \\
        \cdf(y) = 1 - (1-y)^m, \textrm{ for } y \in [0,1].
    \end{align*}
    We are particularly interested in the joint probability density function
    \begin{equation*}
        \pdf (x, y) = n(1-x)^{n-1} m(1-y)^{m-1}, \textrm{ for } (x, y) \in [0,1]^2.
    \end{equation*}
    The probability mass enclosed in a square along the diagonal $ S = [s_1, s_2]^2 \subset [0,1]^2$ is then precisely
    \begin{align}
        \mu(S) &= \int_{s_1}^{s_2} \int_{s_1}^{s_2} n(1-x)^{n-1} m(1-y)^{m-1} dy dx \nonumber \\ &
        = \left[ (1 - s_2)^n - (1-s_1)^n  \right]
           \left[ (1 - s_2)^m - (1-s_1)^m  \right]
           \label{eqn:integral_box}
    \end{align}
    Recall $f_{0,q,r}(A)_0 \in \{1, \ldots, 2^q\} \times [0, 1]^r \equiv \{1, \ldots, 2^q\} \times \\ \{0, \ldots, 2^r -1\}$, so
    given $f_{0,q,r}(A)_0 = (i, j)$, $x = 0.0^00^i1j\cdots$ in the binary expansion, unless $i = 2^q $, in which case the binary expansion is $x = 0.0^00^ij\cdots$.
    That in turn gives $  s_1 <x< s_2$, where $s_1 = \frac{2^r + j}{2^{r+i}}, s_2 = \frac{2^r+j+1}{2^{r+i}}$ when $i < 2^q$, and $s_1 = \frac{j}{2^{r+i-1}}, s_2 = \frac{j+1}{2^{r+i-1}}$.
    Collisions happen precisely when $ s_1 < x, y < s_2$.

    Finally, using the $s_1, s_2$ formulas above, it suffices to sum the probability of collision over the image of $f$, so
    \begin{equation}
        \E \gamma (n, m ) = \sum_{i=1}^{2^q} \sum_{j=0}^{2^r-1} \mu([s_1, s_2]).
    \end{equation}
    Substituting in for $s_1, s_2,$ and $\mu$ completes the proof.
    Note also that this is precisely the sum of the probability mass in the red and purple squares along the diagonal in Figure \ref{fig:hmh-in-practice}.
\end{proof}

\begin{figure}[tb]
    \includegraphics[width=1\columnwidth]{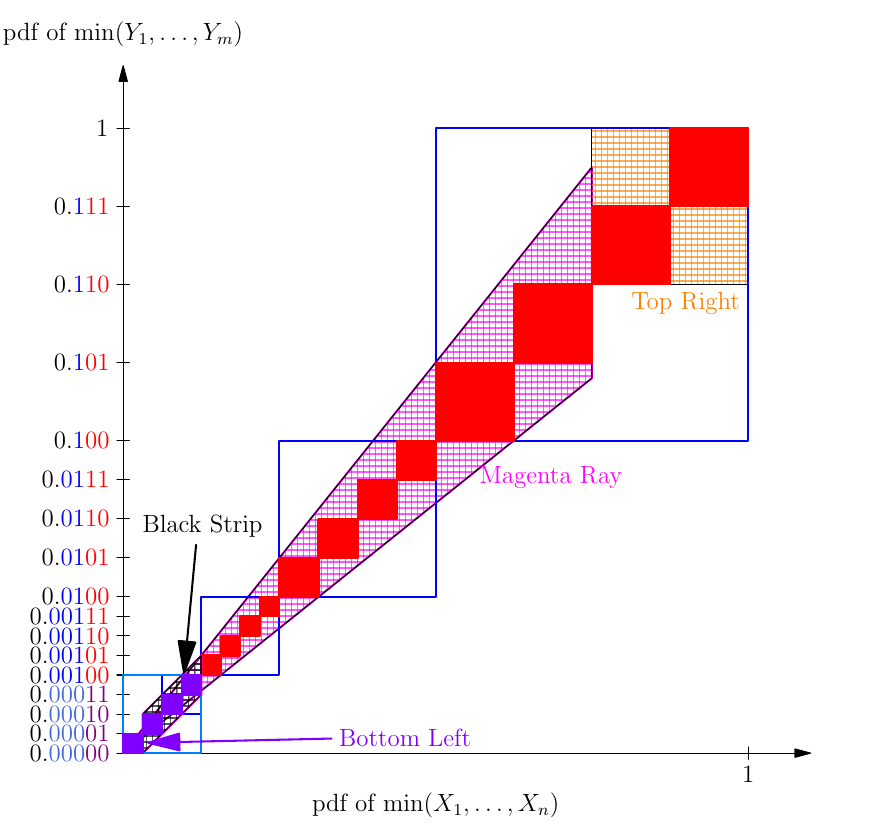}
    \caption{We will upper bound the collision probability of HyperMinHash by dividing it into these four regions of integration: (a) the Top Right orange box, (b), the magenta ray covering intermediate boxes, (c) the black strip covering all but the final purple box, and (d) the final purple sub-bucket by the origin.}
    \label{fig:approximating-hyperminhash}
\end{figure}

While Lemma \ref{thm:collision} allows us to explicitly compute $\E \gamma(m,n)$, the lack of a closed form solution makes reasoning about it difficult.
Here, we will upper bound the expectation by integrating over four regions of the unit square that cover all the collision boxes (Figure \ref{fig:approximating-hyperminhash}).
For ease of notation, let $\hr = 2^r$ and $\hhq = 2^{2^q}$.
\begin{itemize}
    \item The \textbf{T}op \textbf{R}ight box $TR = [\frac{\hr}{\hr+1}, 1]^2$ (in orange in Figure \ref{fig:approximating-hyperminhash}).
    \item The magenta triangle from the origin bounded by the lines $y = \frac{\hr}{\hr+1}x$ and $y = \frac{\hr+1}{\hr}x$ with $0<x<\frac{\hr}{\hr+1}$, which we will denote $RAY$.
    \item The black strip near the origin covering all the purple boxes except the one on the origin, bounded by the lines
        $y = x - \frac{1}{\hr\hhq}$, $y = x + \frac{1}{\hr\hhq}$, and $\frac{1}{\hr\hhq} <x< \frac{1}{\hhq}$, which we will denote $STRIP$.
    \item The \textbf{B}ottom \textbf{L}eft purple box $BL = [0, \frac{1}{\hr\hhq}]^2$.
\end{itemize}

\begin{lemma}
    \label{thm:tr}
    The probability mass contained in the top right square $\mu(TR) \le \frac{1}{\hr}$.
\end{lemma}
\begin{proof}
    By Equation \ref{eqn:integral_box}, 
    \begin{align*}
        \mu(TR) &= 
        \int^{1}_{\frac{\hr}{\hr+1}} \int^{1}_{\frac{\hr}{\hr+1}}  %
            n (1-x)^{n-1} m (1-y)^{m-1} dy dx \\ &
            = \left[ -(1-x)^n  \right]^1_{\frac{\hr}{\hr+1}}
               \left[ -(1-y)^m  \right]^1_{\frac{\hr}{\hr+1}} 
            =\frac{1}{(\hr+1)^{n+m}} \le \frac{1}{\hr}. 
    \end{align*}
\end{proof}

\begin{lemma}
    \label{thm:bl}
    The probability mass contained in the bottom left square near the origin is
$        \mu(BL) \le \frac{n}{\hr\hhq}$.
\end{lemma}
\begin{proof}
    \begin{align*}
        \mu(BL) = &\int_{0}^{\frac{1}{\hr\hhq}} \int_{0}^{\frac{1}{\hr\hhq}}  %
            n (1-x)^{n-1} m (1-y)^{m-1} dy dx  \\
        = & \left[ -(1-x)^n  \right]_0^{\frac{1}{\hr\hhq}}
         \left[ -(1-y)^m  \right]_0^{\frac{1}{\hr\hhq}} \\
        = & \left[ 1- \left(1 - \frac{1}{\hr\hhq}\right)^n  \right]
            \left[ 1- \left(1 - \frac{1}{\hr\hhq}\right)^m  \right].
    \end{align*}
    For $n, m < \hr\hhq$, we note that the linear binomial approximation is actually a strict upper bound (trivially verified through the Taylor expansion), so
        $\mu(BL) \le \frac{nm}{\hr^2\hhq^2} \le \frac{n}{\hr\hhq}$ .
\end{proof}

\begin{lemma}
    \label{thm:ray}
    The probability mass of the ray from the origin can be bounded
    $\mu(RAY) \le \frac{3}{\hr}$ .
\end{lemma}
\begin{proof}
    Unfortunately, the ray is not aligned to the axes, so we cannot integrate $x$ and $y$ separately.
    \begin{align*}
        &\mu(RAY) =   \int^{\frac{\hr}{\hr+1}}_{0} \int^{\frac{\hr+1}{\hr} x}_{\frac{\hr}{\hr+1}x} 
            n (1-x)^{n-1} m (1-y)^{m-1} dy dx  \\
        &=      \int^{\frac{\hr}{\hr+1}}_0 n(1-x)^{n-1} \left[
                \left( 1 - \frac{\hr}{\hr+1}x \right)^m -
                \left( 1 - \frac{\hr+1}{\hr}x \right)^m
                \right] dx
    \end{align*}
    Using the elementary difference of powers formula, note that for $0 \le \alpha \le \beta \le 1$,
    \begin{equation*}
        \alpha^m - \beta^m = (\alpha - \beta) \left(\sum_{i=1}^{m} \alpha^{m-i}\beta^{i-1} \right) \le (\alpha - \beta)m\beta^{m-1} .
    \end{equation*}
    With a bit of symbolic manipulation, we can conclude that
    \begin{align*}
        &\mu(RAY) \\ &\le  \int^{\frac{\hr}{\hr+1}}_0 n(1-x)^{n-1} 
        \left[ \frac{2\hr + 1}{\hr(l+1)} x m \left( 1 - \frac{\hr}{\hr+1}x \right)^{m-1}  \right] dx \\ &
        \le \frac{2\hr + 1}{\hr(l+1)}  \int^{\frac{\hr}{\hr+1}}_0 nm\left(1-\frac{\hr}{\hr+1}\right)^{n + m - 2}  x dx .
    \end{align*}

    With a straight-forward integration by parts,
    \begin{align*}
        &\mu(RAY) \\ \le & - \frac{2\hr+1}{\hr^2} \cdot\frac{nm}{n+m-1} \cdot  \frac{\hr}{\hr+1}\left( 1 - \frac{\hr^2}{(\hr+1)^2} \right)^{n+m-1} \\
           & - \frac{2\hr+1}{\hr^2} \cdot\frac{nm}{n+m-1} \cdot \frac{\hr+1}{\hr} \cdot \frac{1}{n+m}\left( 1 - \frac{\hr^2}{(\hr+1)^2} \right)^{n+m-1} \\
           & + \frac{2\hr+1}{\hr^2} \cdot\frac{nm}{n+m-1} \cdot \frac{\hr+1}{\hr} \cdot \frac{1}{n+m} \\
        \le & \frac{(2\hr+1)(\hr+1)}{\hr^3} \cdot \frac{nm}{(n+m)(n+m-1)} \le \frac{3}{\hr}. \qedhere
    \end{align*}
\end{proof}

\begin{lemma}
    \label{thm:strip}
    The probability mass of the diagonal strip near the origin is
        $\mu(STRIP) \le \frac{2}{\hr}$ .
\end{lemma}
\begin{proof}
    Using the same integration procedure and difference of powers formula used in the proof of Lemma \ref{thm:ray},
    \begin{align*}
        \mu(STRIP) &= \int_{\frac{1}{\hr\hhq}}^{\frac{1}{\hhq}} \int_{x - \frac{1}{\hr\hhq}}^{x + \frac{1}{\hr\hhq}}
            n(1-x)^{n-1}m(1-y)^{m-1} dy dx \\
        &\le \frac{2}{\hr\hhq} \int_{\frac{1}{\hr\hhq}}^{\frac{1}{\hhq}} nm \left( 1 - x + \frac{1}{\hr\hhq} \right)^{n+m-2} \\&
             = \frac{2}{\hr\hhq} \cdot \frac{nm}{n+m-1} \cdot \frac{\hr-1}{\hr\hhq} \le \frac{2}{\hr} .
    \end{align*}
\end{proof}

\begin{proof}[Proof of Proposition \ref{thm:gamma}.]
    Summing bounds from Lemmas \ref{thm:tr}, \ref{thm:bl}, \ref{thm:ray}, and \ref{thm:strip},
        $\E \gamma(n,m) \le \frac{6}{\hr} + \frac{n}{\hr\hhq} = \frac{6}{2^r} + \frac{n}{2^{2^q+r}} .$
\end{proof}

\begin{proof}[Proof of Theorem \ref{thm:main}.]

    Let $A_i$, $B_i$ be the $i$th partitions of $A$ and $B$ respectively. 
    For ease of notation, let us define $\hp = 2^p$.
    Recall that
        $
        C = \sum_{j=0}^{\hp-1} Z_{p,q,r}(A,B,j) .
        $
    We will first bound $\E Z_{p,q,r}(A,B,j)$ using the same techniques used in Proposition \ref{thm:gamma}.
    Notice first that $Z_{p,q,r}(A,B,j)$ effectively rescales the minimum hash values from $Z_{0,q,r}(A,B,j) = \gamma(n,m)$ down by a factor of $2^p$;
    i.e. we scale down both the axes in Figure \ref{fig:approximating-hyperminhash} by substituting $\hhq \gets \hhq \hp$ in
    Lemmas \ref{thm:strip} and \ref{thm:bl}.
    We do not need Lemma \ref{thm:tr} because its box is already covered by the Magenta Ray from Lemma \ref{thm:ray}, which we do not scale.
    Summing these together, we readily conclude
    $    \E Z_{p,q,r}(A,B,j) \le \frac{5}{\hr} + \frac{n}{\hr\hhq\hp} = \frac{5}{2^r} + \frac{n}{2^{p+2^q+r}}$ .
    Then by linearity of expectation,
        $
        \E C \le 
        2^p \left[ \frac{5}{2^r} + \frac{n}{2^{p+2^q+r}} \right] .
        $
\end{proof}

\begin{proof}[Proof of Theorem \ref{thm:variance}.]
    By conditioning on the multinomial distribution, we can decompose $C$ into
    \begin{align*}
        C & = 
        \sum_{\substack{ \alpha_1 + \cdots \alpha_{\hp} = n \\ \beta_1 + \cdots \beta_{\hp} = m } }
             \mathbbm{1}_{\left(\substack{ \forall i, |A_i| = \alpha_i \\ \forall i, |B_i| = \beta_i} \right)}
              \sum_{i=0}^{\hp - 1}
              Z_{p,q,r} \left(A,B,j \biggr| \substack{|A_i|=\alpha_i \\ |B_i|=\beta_i}  \right).
    \end{align*}

    For ease of notation in the following, we will use $\mathring{\alpha}, \mathring\beta$ to denote the event $\forall i, |A_i| = \alpha_i$ and $\forall i, |B_i| = \beta_j$ respectively. Additionally, let $\mathring{Z}(j) = Z_{p,q,r}(A,B,j)$. So,
        $
        C =
            \sum_{\mathring\alpha, \mathring\beta}
            \sum_{j=0}^{\hp-1}
            \mathbbm{1}_{\mathring\alpha, \mathring\beta} 
            \mathring{Z} \left(j \big|\mathring\alpha, \mathring\beta \right) .
            $
    
    Then
    \vspace{-1em}
    \begin{align*}
        & \var(C) = \\  &
            \sum_{\substack{\mathring\alpha_1, \mathring\beta_1 \\
                            \mathring\alpha_2, \mathring\beta_2 }}
            \sum_{j_1, j_2=0}^{\hp-1}
            \cov\left(
                \mathbbm{1}_{\mathring\alpha_1, \mathring\beta_1} 
                \mathring{Z} \left(j_1 \big|\mathring\alpha_1, \mathring\beta_1 \right)
                ,
                \mathbbm{1}_{\mathring\alpha_2, \mathring\beta_2} 
                \mathring{Z} \left(j_2 \big|\mathring\alpha_2, \mathring\beta_2 \right)
            \right) .
    \end{align*}

    But note that for $(\mathring\alpha_1, \mathring\beta_1) \ne (\mathring\alpha_2, \mathring\beta_2)$,
        $
        \mathbbm{1}_{\mathring\alpha_1, \mathring\beta_1} = 1 \implies
        \mathbbm{1}_{\mathring\alpha_2, \mathring\beta_2} = 0
        $
    and vice versa, because they are disjoint indicator variables. As such,
    for $(\mathring\alpha_1, \mathring\beta_1) \ne (\mathring\alpha_2, \mathring\beta_2)$,
    \begin{equation*}
        \cov\left(
            \mathbbm{1}_{\mathring\alpha_1, \mathring\beta_1} 
            \mathring{Z} \left(j_1 \big|\mathring\alpha_1, \mathring\beta_1 \right)
            ,
            \mathbbm{1}_{\mathring\alpha_2, \mathring\beta_2} 
            \mathring{Z} \left(j_2 \big|\mathring\alpha_2, \mathring\beta_2 \right)
        \right) \le 0,
    \end{equation*}
    implying that
    \begin{align*}
        & \var(C) \\ \le  &
            \sum_{\substack{\mathring\alpha, \mathring\beta }}
            \sum_{j_1, j_2=0}^{\hp-1}
            \cov\left(
                \mathbbm{1}_{\mathring\alpha, \mathring\beta} 
                \mathring{Z} \left(j_1 \big|\mathring\alpha, \mathring\beta \right)
                ,
                \mathbbm{1}_{\mathring\alpha, \mathring\beta} 
                \mathring{Z} \left(j_2 \big|\mathring\alpha, \mathring\beta \right)
            \right) \\
        = &
            \sum_{\substack{\mathring\alpha, \mathring\beta }}
            \sum_{j=0}^{\hp-1}
            \var \left(
                \mathbbm{1}_{\mathring\alpha, \mathring\beta} 
                \mathring{Z} \left(j \big|\mathring\alpha, \mathring\beta \right)
            \right) \\ &
        +
            \sum_{\substack{\mathring\alpha, \mathring\beta }}
            \sum_{\substack{j_1 \ne j_2 \\ 0 \le j_1 \le \hp -1 \\ 0 \le j_2 \le \hp - 1}}
            \cov\left(
                \mathbbm{1}_{\mathring\alpha, \mathring\beta} 
                \mathring{Z} \left(j_1 \big|\mathring\alpha, \mathring\beta \right)
                ,
                \mathbbm{1}_{\mathring\alpha, \mathring\beta} 
                \mathring{Z} \left(j_2 \big|\mathring\alpha, \mathring\beta \right)
            \right)  .
    \end{align*}

    Note that the first term can be simplified, recalling that 
    $\mathring{Z}$ is a \{0,1\} Bernouli r.v., so
    \begin{align*}
        &\sum_{j=0}^{\hp -1}
            \var \left(
                \mathbbm{1}_{\mathring\alpha, \mathring\beta} 
                \mathring{Z} \left(j \big|\mathring\alpha, \mathring\beta \right)
            \right)
        \\& = 
        \sum_{j=0}^{\hp -1}
                \var \left( \mathring{Z} \left(j \right) \right) 
        \le
        \sum_{j=0}^{\hp -1}
                \E \left[ \mathring{Z} \left(j \right) \right]  = \E C .
    \end{align*}

    Moving on, from the covariance formula, for independent random variables $X_1, X_2, Y$,
    \begin{align*}
        \cov(X_1 Y, X_2 Y) &= \E[X_1 X_2 Y^2] - \E[X_1 Y] \E [X_2 Y] \\ &
                           = \E[X_1]\E[X_2] \left( \E[Y^2] - \E[Y]^2  \right) \\ &
                           = \E[X_1]\E[X_2] \var(Y)
    \end{align*}
    Thus the second term of the summation can be bounded as follows:
    \begin{align*}
        &\sum_{\substack{\mathring\alpha, \mathring\beta }}
        \sum_{\substack{j_1 \ne j_2 \\ 0 \le j_1 \le \hp -1 \\ 0 \le j_2 \le \hp - 1}}
        \cov\left(
            \mathbbm{1}_{\mathring\alpha, \mathring\beta} 
            \mathring{Z} \left(j_1 \big|\mathring\alpha, \mathring\beta \right)
            ,
            \mathbbm{1}_{\mathring\alpha, \mathring\beta} 
            \mathring{Z} \left(j_2 \big|\mathring\alpha, \mathring\beta \right)
        \right) 
        \\ &
        =
        \sum_{\substack{\mathring\alpha, \mathring\beta }}
        \sum_{\substack{j_1 \ne j_2 \\ 0 \le j_1 \le \hp -1 \\ 0 \le j_2 \le \hp - 1}}
        \E \left[ 
            \mathring{Z} \left(j_1 \big|\mathring\alpha, \mathring\beta \right)
        \right]
        \E \left[
            \mathring{Z} \left(j_2 \big|\mathring\alpha, \mathring\beta \right)
        \right]
        \var \left(
            \mathbbm{1}_{\mathring\alpha, \mathring\beta} 
        \right)
        \\
        &\le
        \sum_{\substack{\mathring\alpha, \mathring\beta }}
        \sum_{j_1 = 0}^{\hp - 1}
        \sum_{j_2 = 0}^{\hp - 1}
        \E \left[ 
            \mathring{Z} \left(j_1 \big|\mathring\alpha, \mathring\beta \right)
        \right]
        \E \left[
            \mathring{Z} \left(j_2 \big|\mathring\alpha, \mathring\beta \right)
        \right]
        \var \left(
            \mathbbm{1}_{\mathring\alpha, \mathring\beta} 
        \right)
        \\ &
        =
        \sum_{\substack{\mathring\alpha, \mathring\beta }}
        \E \left[ C | \mathring\alpha, \mathring\beta  \right]^2
        \var \left(
            \mathbbm{1}_{\mathring\alpha, \mathring\beta} 
        \right)
        \\
        &=
        \sum_{\substack{\mathring\alpha, \mathring\beta }}
        \E \left[ C | \mathring\alpha, \mathring\beta  \right]^2
        \prob \left(
            \mathring\alpha, \mathring\beta
        \right)
        \left( 1 - \prob \left(
            \mathring\alpha, \mathring\beta
        \right)
        \right)
        \\ &
        \le
        \sum_{\substack{\mathring\alpha, \mathring\beta }}
        \E \left[ C | \mathring\alpha, \mathring\beta  \right]^2
        \prob \left(
            \mathring\alpha, \mathring\beta
        \right)
        =
        \E [C]^2
    \end{align*}

    We conclude that
     $   \var (C) \le \E [C]^2 + \E [C]$ .

\end{proof}

\begin{figure*}[tbp]
    \centering
    \includegraphics[width=1\textwidth]{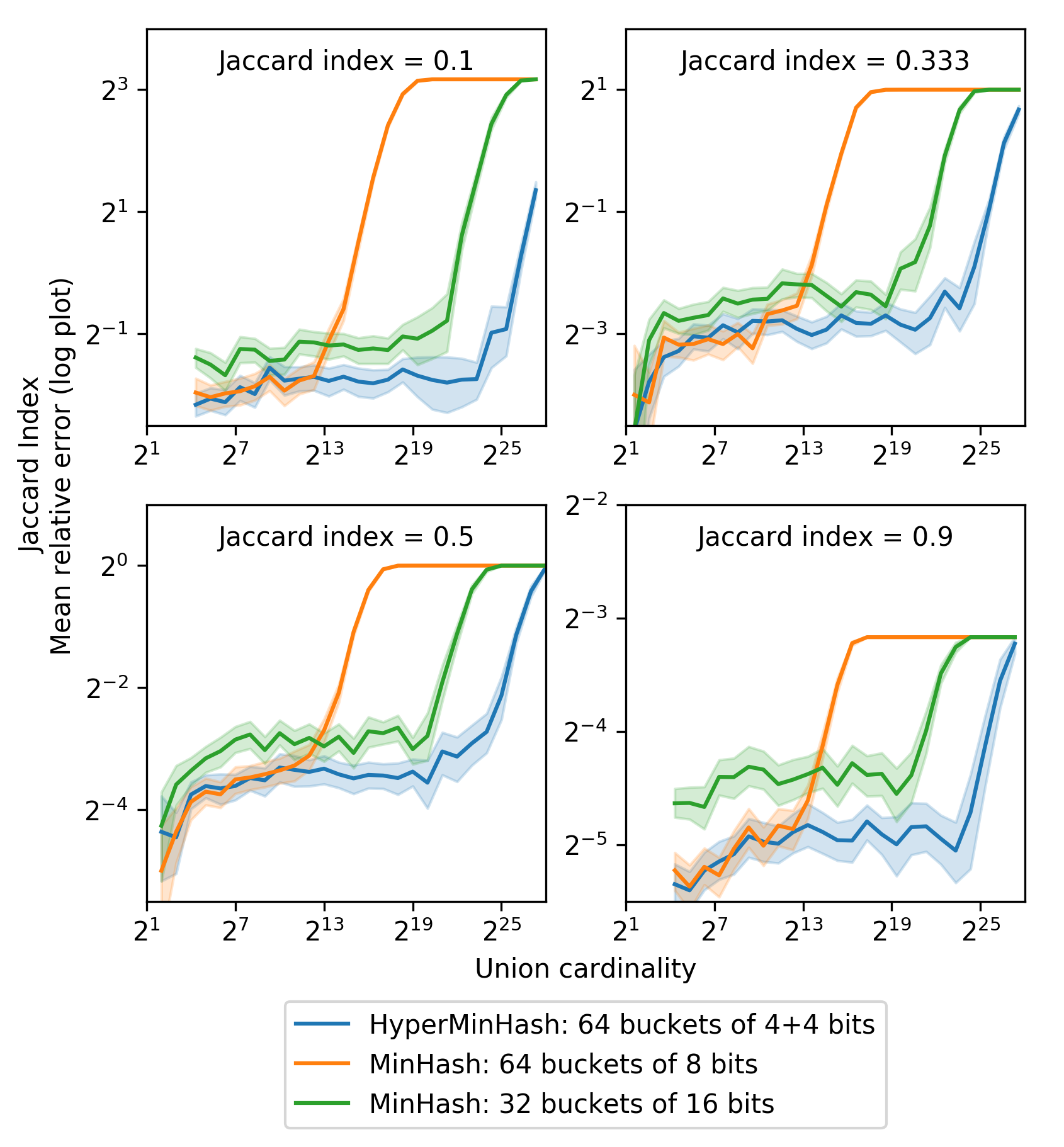}
    \caption{\textbf{For a fixed size sketch, HyperMinHash has better accuracy and/or cardinality range than MinHash.} We compare Jaccard index estimation for identically sized sets with Jaccard index of $0.1$, $1/3$, $0.5$, and $0.9$, and plot the mean relative errors without estimated collision correction. All datasets represent between 8 and 96 random repetitions, as needed to get reasonable 95\% confidence interval (shaded bands). %
    (blue) A 64 byte HyperMinHash sketch, with 64 buckets of 8 bits each, 4 bits of which are allocated to the LogLog counter. Jaccard index estimation remains stable until cardinalities around $2^{21}$. %
    (orange) A 64 byte MinHash sketch with 64 buckets of 8 bits each achieves similar accuracy at low cardinalities, but fails once cardinalities approach $2^{13}$. %
    (green) A 64 byte MinHash sketch with 32 buckets of 16 bits can access larger cardinalities of around $2^{19}$, but to do so trades off on low-cardinality accuracy. }
    \label{fig:experiment}
\end{figure*}

\subsection{Space-complexity analysis}
Note that Theorem~\ref{thm:variance} implies that the standard deviation of the number of collisions is about the same as the number of collisions itself, as bounded in Theorem~\ref{thm:main}.
For Jaccard index $t = J(A,B)$, the absolute error caused by minimum hash collisions is then approximately 
\begin{equation*}
    \frac{\E C}{ 2^p} = \frac{5}{2^r} + \frac{n}{2^{p+2^q+r}}.
\end{equation*}
So long as $n < 2^{2^q}$, the second term is bounded by $1/2^r$.
Then, the absolute error in the Jaccard index estimation from collisions is bounded by $6/2^r$, where the constant $6$ is a gross overestimate (empirically, the constant seems closer to 1).
For any desired relative error $\epsilon$, we then need only let
\begin{equation*}
    r > \log \frac{6}{\epsilon t}.
\end{equation*}
Of course, for relative error $\epsilon$, this then implies that HyperMinHash needs $O(\log\log n + \log(1/\epsilon t))$ bits per bucket.
Additionally, note that $\epsilon^{-2}$ buckets are needed to control the error, completing our proof that HyperMinHash requires
$O\left(\epsilon^{-2} \left( \log\log n + \log \frac{1}{ t  \epsilon} \right)\right)$ 
bits of space.

An astute reader will note that in this space-complexity analysis, we have implicitly assumed that there are no jointly empty buckets in the HyperMinHash sketches of the two sets.
This assumption allows us to use the simpler Jaccard index estimator
\begin{equation*}
    t(A,B) \approx \frac{\textrm{[\# matches]}}{\textrm{[\# buckets]}}.
\end{equation*}
rather than the harder to analyze estimator
\begin{equation*}
    t(A,B) \approx \frac{\textrm{[\# matches]}}{\textrm{[\# buckets]-[\# jointly empty buckets]}}.
\end{equation*}
In the small set regime where there are empty buckets, although our proofs of the number of additional accidental collisions from using HyperMinHash compression still hold, the relative error in the Jaccard index estimator can be higher because there are fewer filled buckets.
However, we note that this does not change the asymptotic results in the large cardinality regime, which we believe of greater interest to practitioners who are considering using HyperMinHash for space-savings (i.e.\ in the small cardinality regime, double logarithmic scaling is of little importance).

\section{Results}
Now that we have presented both detailed algorithms and a theoretical analysis of HyperMinHash, for completeness, we turn to simulated empirical experiments comparing against MinHash.
Given two sets of interest, we first generate one-permutation k-partition MinHash sketches of both.
As detailed in the introduction, this consists of hashing all elements within a set, partitioning them into k buckets, and storing the minimum valued hash within each bucket.
Then, we estimate the Jaccard index by dividing the number of matching buckets by the number of buckets that are not empty in both sketches.
Using HyperMinHash proceeds identically, except that we then compress the minimum valued hash by storing the position of the leading 1 indicator and several bits following it; Jaccard index estimation is identical, except of course that the buckets must match in both the leading 1 indicator position and the additional bits (Figure 1).

In Figure~\ref{fig:experiment}, we allocate 64 bytes for two standard MinHash sketches and a HyperMinHash sketch, and then plot the mean relative error in the Jaccard index for simulated overlapping sets with Jaccard index 0.1, 0.33, 0.5, and 0.9.
Note that for the sake of comparability, we turn off the expected error collision correction of HyperMinHash.
A Python implementation is available on Github (\url{https://github.com/yunwilliamyu/hyperminhash}) for rerunning all experiments and regenerating the figure.

While changing the Jaccard index of the sets changes the baseline mean relative error and the maximum possible mean relative error, the errors for all three methods remain basically stable across the cardinality ranges that they can handle without hash collisions.
We ran HyperMinHash with 64 buckets of 8 bits, with 4 allocated to the leading 1 indicator.
This sketch was able to access cardinalities ranging up to $2^{21}$.
When running MinHash with 64 buckets of 8 bits, it maintains the same error level as HyperMinHash for union cardinalities below $2^{13}$, but then fails.
Using fewer buckets with larger numbers of bits, such as when we ran MinHash with 32 buckets of 16 bits, allows it to access higher cardinalities (around $2^{19}$), but at the cost of a higher baseline error level.
Thus, for fixed sketch size and cardinality range, HyperMinHash is more accurate; or, for fixed sketch size and bucket number, HyperMinHash can access exponentially larger set cardinalities.

\section{Conclusion}
We have introduced HyperMinHash, a compressed version of the MinHash sketch in $\log\log$ space, and made available a prototype Python implementation at \url{https://github.com/yunwilliamyu/hyperminhash}.
It can be thought of as a compression scheme for MinHash that reduces the number of bits per bucket to $\log\log(n)$ from $\log(n)$ by using insights from HyperLogLog and k-partition MinHash.
As with the original MinHash, it retains variance on the order of $k/t$, where $k$ is the number of buckets and $ t $ is the Jaccard index between two sets.
However, it also introduces $1/l^2$ variance, where $l = 2^r$, because of the increased number of collisions.

Alternately, it can be thought of as an extension for sublogarithmic Jaccard index fingerprinting methods such as b-bit MinHash \cite{li2010b}, adding back in the features of streaming updates and union sketching.
Should a practitioner desire only to create a one-time Jaccard index fingerprint of a set, we recommend they use b-bit MinHash; however, we believe that HyperMinHash serves as a better all-around drop-in replacement of MinHash because it preserves more of MinHash's compositional features than other more space-efficient Jaccard index fingerprints.

There remain many theoretical improvements to be made, especially in removing the requirement of a random oracle.
Notably, the KNW sketch \cite{kane2010optimal} sketch improves HyperLogLog to using constant size buckets and a double logarithmic offset, while also using limited randomness.
We envision that their techniques can be extended to HyperMinHash to further reduce space complexity, as the HyperLogLog parts of the buckets are identical to regular HyperLogLog.
Similarly, Feigenblat, et al.~introduce d-k-min-wise hash functions to reduce the amount of necessary randomness for a bottom-k MinHash sketch \cite{feigenblat2017dk}.
Although we have analyzed HyperMinHash as a k-partition variant, the same floating-point encoding can be applied to bottom-k sketches, with similar error behavior.

Luckily, even without these theoretical improvements, HyperMinHash is already practically applicable. For reasonable parameters of $p = 15, q=6, r=10$, the HyperMinHash sketch will use up 64KiB memory per set, and allow for estimating Jaccard indices of 0.01 for set cardinalities on the order of $10^{19}$  with accuracy around 5\%.
HyperMinHash is to our knowledge the first practical streaming summary sketch capable of directly estimating union cardinality, Jaccard index, and intersection cardinality in $\log\log$ space, able to be applied to arbitrary Boolean formulas in conjunctive normal form with error rates bounded by the final result size.
We hope that HyperMinHash as presented in this manuscript will be of utility for Boolean queries on large databases.

\ifCLASSOPTIONcompsoc
  \section*{Acknowledgments}
\else
  \section*{Acknowledgment}
\fi

    This study was supported by National Institutes of Health (NIH) Big Data to Knowledge (BD2K) awards U54HG007963 from the National Human Genome Research Institute (NHGRI) and U01CA198934 from the National Cancer Institute (NCI). The funders had no role in study design, data collection and analysis, decision to publish, or preparation of the manuscript.
    We thank Daphne Ippolito and Adam Sealfon for useful comments and advice, and Seif Lotfy for finding a minor bug in our pseudo-code.
	We also are indebted to all of our reviewers, who kindly but firmly pointed out missing references and hidden assumptions in a previous draft of this manuscript.

\ifCLASSOPTIONcaptionsoff
  \newpage
\fi

\bibliographystyle{IEEEtran}
\bibliography{IEEEabrv,main}

\begin{thebibliography}{10}
\providecommand{\url}[1]{#1}
\csname url@samestyle\endcsname
\providecommand{\newblock}{\relax}
\providecommand{\bibinfo}[2]{#2}
\providecommand{\BIBentrySTDinterwordspacing}{\spaceskip=0pt\relax}
\providecommand{\BIBentryALTinterwordstretchfactor}{4}
\providecommand{\BIBentryALTinterwordspacing}{\spaceskip=\fontdimen2\font plus
\BIBentryALTinterwordstretchfactor\fontdimen3\font minus
  \fontdimen4\font\relax}
\providecommand{\BIBforeignlanguage}[2]{{%
\expandafter\ifx\csname l@#1\endcsname\relax
\typeout{** WARNING: IEEEtran.bst: No hyphenation pattern has been}%
\typeout{** loaded for the language `#1'. Using the pattern for}%
\typeout{** the default language instead.}%
\else
\language=\csname l@#1\endcsname
\fi
#2}}
\providecommand{\BIBdecl}{\relax}
\BIBdecl

\bibitem{broder1997resemblance}
A.~Z. Broder, ``On the resemblance and containment of documents,'' in
  \emph{Compression and Complexity of Sequences 1997. Proceedings}.\hskip 1em
  plus 0.5em minus 0.4em\relax IEEE, 1997, pp. 21--29.

\bibitem{li2005using}
P.~Li and K.~W. Church, ``Using sketches to estimate associations,'' in
  \emph{Proceedings of the conference on Human Language Technology and
  Empirical Methods in Natural Language Processing}.\hskip 1em plus 0.5em minus
  0.4em\relax Association for Computational Linguistics, 2005, pp. 708--715.

\bibitem{jaccard1902lois}
P.~Jaccard, ``Lois de distribution florale dans la zone alpine,'' \emph{Bull
  Soc Vaudoise Sci Nat}, vol.~38, pp. 69--130, 1902.

\bibitem{bar2002counting}
Z.~Bar-Yossef, T.~Jayram, R.~Kumar, D.~Sivakumar, and L.~Trevisan, ``Counting
  distinct elements in a data stream,'' in \emph{International Workshop on
  Randomization and Approximation Techniques in Computer Science}.\hskip 1em
  plus 0.5em minus 0.4em\relax Springer, 2002, pp. 1--10.

\bibitem{durand2003loglog}
M.~Durand and P.~Flajolet, ``Loglog counting of large cardinalities,'' in
  \emph{European Symposium on Algorithms}.\hskip 1em plus 0.5em minus
  0.4em\relax Springer, 2003, pp. 605--617.

\bibitem{flajolet2004counting}
P.~Flajolet, ``Counting by coin tossings,'' in \emph{Advances in Computer
  Science-ASIAN 2004. Higher-Level Decision Making}.\hskip 1em plus 0.5em minus
  0.4em\relax Springer, 2004, pp. 1--12.

\bibitem{flajolet2007hyperloglog}
P.~Flajolet, {\'E}.~Fusy, O.~Gandouet, and F.~Meunier, ``Hyperloglog: the
  analysis of a near-optimal cardinality estimation algorithm,'' in \emph{AofA:
  Analysis of Algorithms}.\hskip 1em plus 0.5em minus 0.4em\relax Discrete
  Mathematics and Theoretical Computer Science, 2007, pp. 137--156.

\bibitem{kane2010optimal}
D.~M. Kane, J.~Nelson, and D.~P. Woodruff, ``An optimal algorithm for the
  distinct elements problem,'' in \emph{Proceedings of the twenty-ninth ACM
  SIGMOD-SIGACT-SIGART symposium on Principles of database systems}.\hskip 1em
  plus 0.5em minus 0.4em\relax ACM, 2010, pp. 41--52.

\bibitem{bachrach2010fast}
Y.~Bachrach and E.~Porat, ``Fast pseudo-random fingerprints,'' \emph{arXiv
  preprint arXiv:1009.5791}, 2010.

\bibitem{li2010b}
P.~Li and C.~K{\"o}nig, ``b-bit minwise hashing,'' in \emph{Proceedings of the
  19th international conference on World wide web}.\hskip 1em plus 0.5em minus
  0.4em\relax ACM, 2010, pp. 671--680.

\bibitem{bachrach2015fingerprints}
Y.~Bachrach and E.~Porat, ``Fingerprints for highly similar streams,''
  \emph{Information and Computation}, vol. 244, pp. 113--121, 2015.

\bibitem{li2012one}
P.~Li, A.~Owen, and C.-H. Zhang, ``One permutation hashing,'' in \emph{Advances
  in Neural Information Processing Systems}, 2012, pp. 3113--3121.

\bibitem{thorup2013bottom}
M.~Thorup, ``Bottom-k and priority sampling, set similarity and subset sums
  with minimal independence,'' in \emph{Proceedings of the forty-fifth annual
  ACM symposium on Theory of computing}.\hskip 1em plus 0.5em minus 0.4em\relax
  ACM, 2013, pp. 371--380.

\bibitem{feigenblat2011exponential}
G.~Feigenblat, E.~Porat, and A.~Shiftan, ``Exponential time improvement for
  min-wise based algorithms,'' in \emph{Proceedings of the twenty-second annual
  ACM-SIAM symposium on Discrete Algorithms}.\hskip 1em plus 0.5em minus
  0.4em\relax SIAM, 2011, pp. 57--66.

\bibitem{feigenblat2017dk}
------, ``dk-min-wise independent family of hash functions,'' \emph{Journal of
  Computer and System Sciences}, vol.~84, pp. 171--184, 2017.

\bibitem{cohen2016min}
E.~Cohen, ``Min-hash sketches.'' 2016.

\bibitem{pagh2014min}
R.~Pagh, M.~St{\"o}ckel, and D.~P. Woodruff, ``Is min-wise hashing optimal for
  summarizing set intersection?'' in \emph{Proceedings of the 33rd ACM
  SIGMOD-SIGACT-SIGART symposium on Principles of database systems}.\hskip 1em
  plus 0.5em minus 0.4em\relax ACM, 2014, pp. 109--120.

\bibitem{ertl2017new}
O.~Ertl, ``New cardinality estimation algorithms for hyperloglog sketches,''
  \emph{arXiv preprint arXiv:1702.01284}, 2017.

\bibitem{ertl2017new2}
------, ``New cardinality estimation methods for hyperloglog sketches,''
  \emph{arXiv preprint arXiv:1706.07290}, 2017.

\bibitem{cohen2017hyperloglog}
E.~Cohen, ``Hyperloglog hyperextended: Sketches for concave sublinear frequency
  statistics,'' in \emph{Proceedings of the 23rd ACM SIGKDD International
  Conference on Knowledge Discovery and Data Mining}.\hskip 1em plus 0.5em
  minus 0.4em\relax ACM, 2017, pp. 105--114.

\end{thebibliography}

\begin{IEEEbiography}[{\includegraphics[trim={15px 0 10px 0}, width=1in,height=1.25in, clip=true, keepaspectratio]{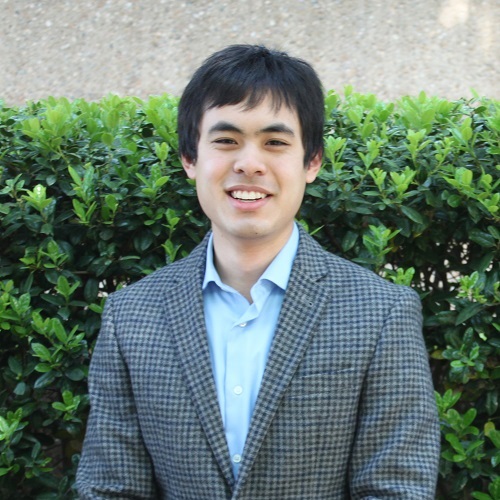}}   ]{Yun William Yu}
    is an assistant professor at the Department of Mathematics, University of Toronto.
\end{IEEEbiography}
\begin{IEEEbiography}[{\includegraphics[ width=1in,height=1.25in, clip,  keepaspectratio]{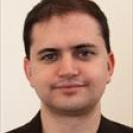}}   ]{Griffin M. Weber}
    is an associate professor at the Department of Biomedical Informatics, Harvard Medical School.
\end{IEEEbiography}

\end{document}